\documentclass{article}
\usepackage{arxivnopreprint}

\usepackage{amsmath}
\usepackage{amsthm}
\usepackage{graphicx}          
\usepackage{amsmath} 
\usepackage{amssymb}  
\usepackage{color,soul}
\usepackage{enumitem}
\usepackage{flushend}
\usepackage{cite}

\allowdisplaybreaks
\newtheorem{thmm}{Theorem}

\newtheorem{remm}{Remark}
\newtheorem{assm}{Assumption}
\newtheorem{deff}{Definition}
\newtheorem{Lemm}{Lemma}

\DeclareMathOperator*{\argmax}{argmax}

\begin{document}

\title{Route Choice-based Socio-Technical Macroscopic Traffic Model}

\author{
    Tanushree Roy \\
    Department of  Mechanical Engineering \\
     The Pennsylvania State University \\
    University Park, PA 16802, USA. \\
    \texttt{tbr5281@psu.edu}\\
    \And
    Satadru Dey \\
     Department of  Mechanical Engineering \\
     The Pennsylvania State University \\
    University Park, PA 16802, USA. \\
    \texttt{ skd5685@psu.edu}
}

\maketitle

\begin{abstract}
Human route choice is undeniably one of the key contributing factors towards traffic dynamics. However, most existing macroscopic traffic models are typically concerned with driving behavior and do not incorporate human route choice behavior models in their formulation. In this paper, we propose a socio-technical macroscopic traffic model that characterizes the traffic states using human route choice attributes. Essentially, such model provides a framework for capturing the Cyber-Physical-Social coupling in smart transportation systems. To derive this model, we first use Cumulative Prospect Theory (CPT) to model the human passengers' route choice under bounded rationality. These choices are assumed to be influenced by traffic alerts and other incomplete traffic information. Next, we assume that the vehicles are operating under a non-cooperative cruise control scenario. Accordingly, human route choice segregates the traffic into multiple classes where each class corresponds to a specific route between an origin-destination pair. Thereafter, we derive a Mean Field Game (MFG) limit of this non-cooperative game to obtain a macroscopic model which embeds the human route choice attribute. Finally, we analyze the mathematical characteristics of the proposed model and present simulation studies to illustrate the model behavior.
\end{abstract}

\maketitle

\section{Introduction}

Modern Intelligent Transportation Systems (ITSs) exhibit strong interactions between Information \& Communication Technology (ICT), physical traffic flow, and human social behavior. Such interaction is even stronger in smart mobility solutions such as Connected  Adaptive Cruise Control (CACC) systems. This motivates the need for modeling traffic dynamics in ITSs as socio-technical systems which capture the Cyber-Physical-Social coupling \cite{whitworth2013social}. Essentially, human cognitive and social behavior is incorporated with the physical vehicular dynamics and cybernetic strategies of the smart mobility solutions in this type of socio-technical models. Such modeling strategies have the potential to provide a quantifiable connection between smart mobility dynamics and human behavioral dynamics.  Along this line, this paper presents a socio-technical modeling framework for traffic systems that characterizes the macroscopic traffic dynamics in terms of the decision behavior of the human-in-the-vehicle, under a CACC driving scenario. 

In literature there exist different frameworks for modeling human behavioral aspects, Cumulative Prospect Theory (CPT) being one of them \cite{CPT_main}. Within this framework, the behavior of humans as decision makers under uncertainty is of bounded rationality. Essentially, CPT models this decision behavior using subjective utility of outcome and subjective perception of probability. CPT was originally proposed in the context of economics, and later applied to human-in-the-loop technical systems. For example, CPT is employed in consumer-behavior based electricity pricing \cite{smartgrid_CPT}, cloud-storage defense strategy \cite{cloud_defense_CPT}, and evaluation of renewable power sources \cite{eval_CPT}. In transportation related applications, CPT has been explored in  passenger behavior modeling based on waiting time \cite{Wait_TimeCPT}, driver behavior in context  changing to High-Occupancy-Vehicle lanes \cite{HOV_lanechange_CPT}, interaction framework between traffic information provider and user modeled as a Stackelberg game \cite{CPT_stakelberg}, developing adaptive pricing strategies for Shared Mobility on Demand Services \cite{AnnaswamyCPT}. Besides, CPT has been explored in the context of modeling driver route choice behavior where the traveller has incomplete traffic information \cite{xu_routeCPT}, friends' travel information \cite{zhang_frnd_cumulative}, and traffic information through variable message sign indicators \cite{gao_route,gan_msg_sign}. In \cite{travel_info_impact},  Logit Kernel to model human route choice as a function of travel time information. Nevertheless, these aforementioned works do not to explore the impact of such human route choices on macroscopic traffic behavior. This is especially relevant now because of the increasing  distribution of traffic information through platforms such Google Maps and INRIX. 


Existing traffic modeling strategies can be generally categorized into two classes: (i) vehicle-based or microscopic, and (ii) traffic flow-based or macroscopic. In literature, microscopic modeling of Adaptive Cruise Control (ACC) enabled cars in the cooperative setting has been explored in \cite{milanes_CACC,xiao_CACC_micro,davis2004effect,wang2013coop,wang2014coop} while ACC under non-cooperative setting has been addressed in \cite{talebpour2015modeling,wang2014noncoop, yu2018noncoop}. The inherent disadvantage of such microscopic models is the computational burden for increased number of vehicles. Hence, efforts have been made towards macroscopic modeling of ACC enabled traffic flow in cooperative \cite{ngoduy2013instability, delis2015macroscopic} and non-cooperative \cite{nikolos_CACC_macroscopic,Di2019game,Di2} settings. Additionally, in \cite{sadigh2016planning,chicken,li2017game}, interactions between individual autonomous vehicle and human-driven vehicle or pedestrian have been investigated. {In \cite{Di2019game}, a macroscopic model traffic flow model for autonomous vehicles has been derived using Mean Field Games (MFG) by connecting microscopic vehicular dynamics to macroscopic traffic flow. Additionally, the same authors extended their work in \cite{Di2} to capture traffic dynamics in a mixed traffic scenario containing autonomous vehicles and human-driven vehicles. However, in \cite{Di2019game,Di2}, the human-driven vehicles lack human behavioral models and autonomous vehicles were modelled as rational agents. Hence, these works do not address the impact of route choices made by the human passengers on macroscopic traffic flow. In our present paper, we utilize MFG setting (similar to the one discussed in \cite{Di2019game}) to connect microscopic dynamics to macroscopic models, and propose a modeling framework to capture the human behavioral aspects in macroscopic traffic dynamics. Such incorporation of human behavioral aspects enables understanding of traffic dynamics in realistic settings with cyber-physical-social interactions. This is essential since even with fully autonomous vehicles, some features such as choice of routes will often lie with the human passengers \cite{route_choice_withAV}. Evidently, these choices would be influenced by human behavior and the incomplete traffic information available to them. }

To address the aforementioned gap, the main contribution of the paper is a socio-technical macroscopic traffic flow model with human route choice attributes. Specifically, we formulate a multi-class model for non-cooperative CACC enabled vehicle traffic flow, which is parameterized by outputs from human route choice behavior model. We have modelled the human decision making behavior using CPT where the utility of choosing a specific route depends on the incoming traffic information and driving convenience knowledge about each route, e.g. general road conditions and presence of tolls. Needless to say, this route choice behavior dictates the number vehicle along each route. Next, we have assumed that all vehicles are non-cooperatively optimizing their cost functional along the chosen route which in turn leads to a differential game setting. In order to transition from this microscopic setting to macroscopic traffic characteristics, we take Mean Field Game (MFG) limit of this differential game. Subsequently, we obtain a multi-class model for the macroscopic traffic flow that provides the continuum equation as well as the dynamics of the driving cost. Here, this driving cost embeds the human route choice behavior as an attribute of the traffic flow. We also analyze the following mathematical properties of the socio-technical traffic model: (i) Fundamental Diagram, (ii) hyperbolicity, and (iii) Lyapunov-based linear stability.

This paper is organized as follows. Section II discusses the Problem statement of the paper, Section III develops the modeling framework, Section IV proposed the Socio-technical model for traffic, Section V discusses a case study, Section VI shows the simulation results of our work followed by conclusion in  Section VII.

\section{Problem Statement}
We consider a macroscopic traffic flow of CACC enabled connected and autonomous vehicles between single Origin-Destination (OD) pair A and C as shown in \mbox{Figure \ref{fig:route}}. Without loss of generality, we consider two routes between this OD pair that fork at junction C, as shown in \mbox{Figure \ref{fig:route}}. At junction C, each vehicle has to choose either Route 1 or 2. Such choice is made by the human passenger in the vehicle based on the accessible route information. Essentially, the traffic flow starting at Origin A divides into two streams at junction C, where one stream takes Route 1 and the other stream takes Route 2. Next, we make the following assumptions regarding our problem setting.

\begin{figure}[ht]
    \centering
    \includegraphics[trim = 0mm 0mm 0mm 0mm, clip, scale=1.0, width=0.5\linewidth]{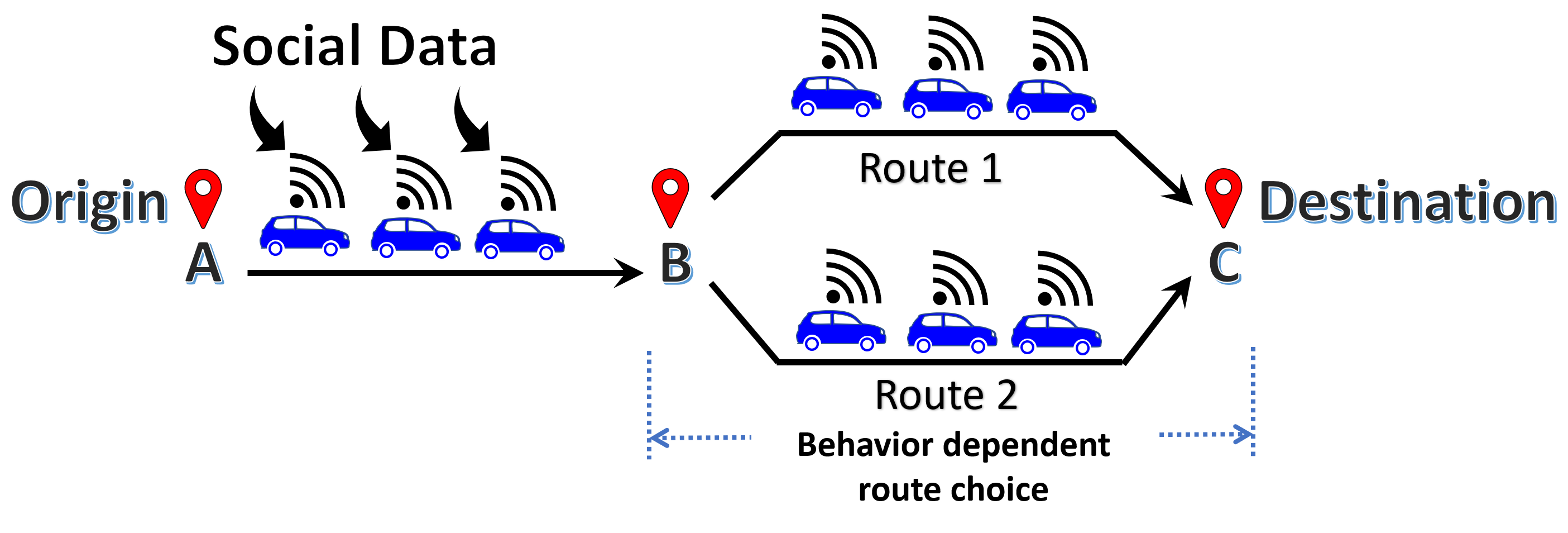}
    \caption{Socio-Technical traffic dynamics.}
    \label{fig:route}
\end{figure}

\begin{assm}
Each vehicle is occupied by a passenger who has control over vehicle route selection.
\end{assm}

\begin{assm}
Routes 1 and 2 have different characteristics in terms of estimated travel time, maximum density, and road condition. As expected, these factors influence the decision of the human passenger in the vehicle to opt for a specific route. Additionally, we assume that the condition of these routes are dynamically changing with time. For example, such changes can occur due to accidents, unpredictable road conditions from construction work or pedestrian traffic. 
\end{assm}

\begin{assm}
The human passenger in the vehicle has partial information about changing characteristics of each route through various traffic reporting platforms such as Google Maps, MapQuest, INRIX, twitter feeds, and cell phone texts. Such an assumption is reasonable in light of the extensive dissemination of social data in today's world \cite{eagle2014reality}.
\end{assm}

\begin{assm}
Due to the presence of partial or incomplete information, the human passenger is assumed to be a bounded rational agent where their preferences are motivated to maximize the \textit{utility} of a choice. 
\end{assm}

\begin{assm}
All the vehicles are identical in terms of physical structure and CACC driving capabilities. While behavior of vehicles choosing a specific route are homogeneous within the group, groups of vehicle corresponding to each route differ from one another in terms of the driving costs. In other other words, depending on the choice of route the vehicles are grouped into classes that are dependent on the costs to drive on that particular route. 
\end{assm}

\begin{assm}
All vehicles receive certain traffic updates through traffic reporting platforms at the same time. 
\end{assm}

Based on the aforementioned setting, our goal is to model the macroscopic traffic dynamics in the whole segment A-B-C, characterized by the behavioral dynamics of the human passengers. To develop this modeling strategy, we execute the following steps. A schematic of the model development framework is shown in Fig. \ref{fig:model}.

\begin{figure}[ht]
    \centering
    \includegraphics[trim = 0mm 0mm 0mm 0mm, clip, scale=1.0, width=0.5\linewidth]{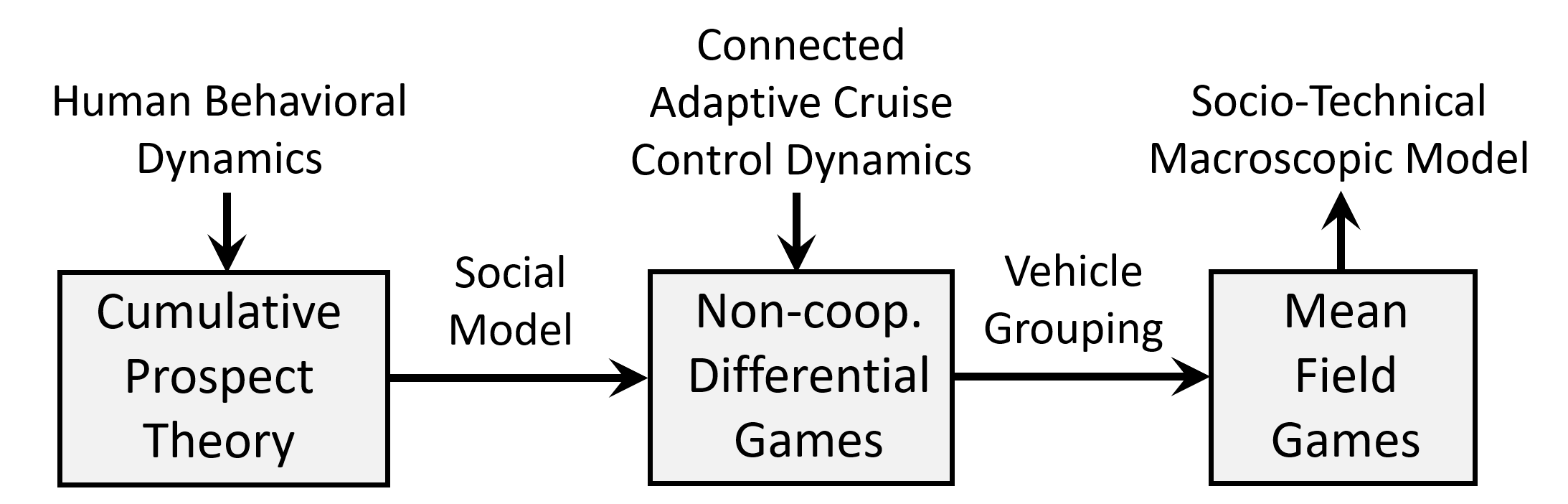}
    \caption{Socio-Technical macroscopic traffic modeling framework.}
    \label{fig:model}
\end{figure}

\textit{Step 1:} First, we model the human choice dynamics using Cumulative Prospect Theory (CPT) which takes into account human behavioral traits such as loss aversion, perception of loss or gain of an outcome dependent on some reference and under-weighing common events while over-weighing unlikely events. Depending on the probability of each human passenger to choose a particular route, the number of vehicles that will proceed towards Route $j \in \{1,2\}$ is decided.

\textit{Step 2:} Next, utilizing these two groups of vehicles (corresponding to two routes), a non-cooperative differential game is set up such that (i) vehicles in each group try to obtain a minimum driving cost, and (ii) vehicles in one group maintain safety from other group of vehicles. 

\textit{Step 3:} Subsequently, we take an arbitrarily large limit for each group of vehicles to derive a Mean Field Games (MFG). This generates a multi-class macroscopic traffic flow model of CACC-enabled vehicles. This multi-class macroscopic model is a socio-technical traffic model since human route choice attribute  parameterizes its Fundamental Diagram.

In the next section, we discuss these steps in detail.

\section{Model Development Framework}

In this section, we detail the model development framework. We start with CPT based modeling of human choice behavior followed by differential and mean field games based macroscopic models.

\subsection{Modeling Route Choice Based on  Traffic Alerts}

The route selection decision of human passenger under the influence of traffic alerts can be modeled using Cumulative Prospect theory (CPT). The behavior of humans as decision makers under uncertainty is of bounded rationality \cite{gigerenzer_bounded_rational}. CPT models this decision behavior using subjective utility of outcome and subjective perception of probability \cite{CPT_main}. The prospect value of choosing a route is uncertain due to various reasons such as congestion, unpredictable road conditions among many others. As described before, we assume two  possible prospects: either someone chooses Route 1 or Route 2. We intend to utilize CPT to obtain the probability of people who would act favourably on receiving a traffic alerts suggesting a change from Route 1 to 2. In other words, we compute the fraction of people who would change their route from Route 1 to 2 upon receiving a route change recommendation. This fraction depends on the subjective perception of the prospect value of a choice and can be modeled using CPT in the following way.

A prospect $j$ is denoted by a sequence of pairs consisting of utility or value (of loss or gain) and probability (of loss or gain), for $N_j$ possible outcomes i.e. $(z^j_1,p^j_1,\hdots,z^j_{N_j},p^j_{N_j})$ where $z^j_i$ are the utilities (modeled as discrete random variables) and $p^j_i$ are the corresponding probabilities with $i \in \{1,\cdots,N_j\}$. Consequently, the objective value of a prospect is given by \cite{CPT_main}
\begin{align}
\mathbb{U}_o^j=\sum_{i=1}^{N_j}z^j_ip^j_i.
\end{align}
In our problem set-up, there are two prospects that the human decision maker is choosing from: Route 1 and Route 2. Hence in our case $j\in \{1,2\}$. 
 
However, human decisions are most often subjective and far from rational. That is, the utility and probability of each outcome is perceived differently. CPT formulates this modification of utility and probability under the following axioms of human decision phenomenon \cite{CPT_main}.

\textbf{Subjective Utility:} The perceived gain or loss value of an outcome is affected by the following factors:
    \begin{itemize}
    \item \textit{Framing Effect:} Subjective loss or gain from a prospect is perceived with respect to a reference value.
    \item \textit{Loss Aversion:} Humans are more affected by a loss than an equal amount of gain. This causes a difference in how loss or gain of a prospect is perceived subjectively and leads to an attitude of risk aversion in face of loss outcome.
    \item \textit{Diminishing Value Sensitivity:} Individuals are less affected by \textit{ changes} in loss (gain) when the value of the prospect is already in high losses (high gains). 
    \end{itemize}

\textbf{Probability Distortion:} The perception of probability of an outcome is affected by the following factors:
    \begin{itemize}
    \item \textit{Over/Underweighting:} Human decision behavior is highly influenced by unlikely events while ignoring highly probable events i.e. low probability events are over-weighted compared to highly probability events. 
    \item \textit{Increased Probability Sensitivity:} Near the end points of probability $p=0$ and $p=1$, changes in probability is perceived more than for mid-range probability.
    \item \textit{Relative sensitivity:} Different attitudes associated towards probabilities for gain and losses are observed and are listed as follows: (i) Risk aversion over gains of high probability, (ii) Risk aversion over losses of small probability, (iii) Risk seeking  over gains of low probability and (iv) Risk seeking over losses of high probability.
    \end{itemize}

Motivated by these axioms, the subjective utility of an outcome can be modeled by the value (or utility) function $z \to U(z)$ given by
\begin{align}\label{utility_func}
    U(z)=\left\{\begin{array}{cc}
        (z-R)^{\beta^+}, & \text{if } z> R, \\
        -\lambda(R-z)^{\beta^-}, & \text{otherwise,}
    \end{array}  \right.
\end{align}
where $R$ is the reference for loss or gain, $\lambda>1$ denotes loss aversion, and $\beta^+>0$ and $\beta^-<1$ denote the diminishing sensitivities to loss or gain. 

Let the utilities be arranged in an ascending order from maximum loss to maximum gain, i.e. $z^j_1\leqslant z^j_2 \leqslant \hdots\leqslant z^j_g \leqslant 0\leqslant z^j_{g+1} \leqslant \hdots \leqslant z^j_{N_j}$ for a given prospect $j$. Here positive utilities imply gains and non-positive utilities imply losses. Let the probability corresponding to these utilities are given by $p^j_1,p^j_2,\hdots, p^j_g,p^j_{g+1},\hdots, p^j_{N_j}$. Then the probability distortion experienced by a limited rational human decision maker is given by the probability weighting function $w(p)$ defined as
\begin{align}
    w(p)=e^{-(-\ln p)^\gamma}.
\end{align}
The function $w(.)$ is 
known as the Prelec's function \cite{prelec}. For $0<\gamma<1$, it produces the characteristic inverse-S shaped curve showing the desired probability distortion characteristics. In order to capture the four risk seeking or avoiding characteristics (mentioned under \textit{Relative Sensitivity} above), CPT defines $\pi_1^{j-},\cdots,\pi_g^{j-}$ as the transformation of the probability for loss, and $\pi^{j+}_{g+1},\cdots,\pi_{N_j}^{j+}$ as the transformation of the probability for gain, which are given by \cite{fennema1997original}
\begin{align}
    & \pi_1^{j-}=w(p^j_1), \\\label{piplus}
    & \pi^{j-}_i=w\left(\sum_{m=1}^ip^j_m\right)-w\left(\sum_{m=1}^{i-1}p^j_m\right),   
\end{align}    
where $2\leqslant i \leqslant g$, and 
\begin{align}
    & \pi_{N_j}^{j+}=w(p^j_{N_j}), \\\label{piminus} &\pi^{j+}_i=w\left(\sum_{m=i}^{N_j}p^j_m\right)-w\left(\sum_{m=i+1}^{N_j}p^j_m\right), 
\end{align}
where $g+1\leqslant i \leqslant N_j-1$. The transformations $\pi_1^{j-},\cdots,\pi_g^{j-}, \pi^{j+}_{g+1},\cdots,\pi_{N_j}^{j+}$ essentially are the decision weights for each outcome $z^j_i$. Subsequently, the overall CPT (subjective) value of a prospect $j$ is given by
\begin{align}\label{discrete_U}
    \mathbb{U}_j=\sum\limits_{i=1}^{g}\pi_i^{j-} U(z^j_i)+\sum\limits_{i=g+1}^{N_j}\pi_i^{j+} U(z^j_i).
\end{align}
Next, consider the case where we have utility $\mathbb{Z}$ as a continuous random variable as opposed to discrete utility variables $z^j_i, i \in \{1,\cdots,N_j\}$. The probability and complimentary distribution functions of this continuous random utility $\mathbb{Z}$ are related to the discrete probability distribution by the following:
\begin{align}
    & F_\mathbb{Z}(z^j_i)=P(\mathbb{Z}\leqslant z^j_i)=\sum\limits_{z^j_g\leqslant z^j_i}p(z^j_g)=\sum_{g=1}^ip^j_g,\\
    & F_\mathbb{Z}(z^j_{i-1})=P(\mathbb{Z}>z^j_{i-1})=\sum\limits_{z^j_g>z_{i-1}}p(z^j_g)=\sum_{g=i}^{N_j}p^j_g.
\end{align} 

Accordingly, the overall CPT (subjective) value of a prospect $j$ \eqref{discrete_U} can be re-written as 
\begin{align}\nonumber
    \mathbb{U}_j=&\int_{-\infty}^RU(z)\frac{d}{dz}\left\{\pi(F_\mathbb{Z}(z))\right\}dx+\int_R^\infty U(z)\frac{d}{dz}\left\{-\pi(1-F_\mathbb{Z}(z))\right\}dx.
\end{align}
where $R$ is the reference for loss or gain, defined in \eqref{utility_func}. That is, $R$ is the certain level of utility that the human passenger perceives as needed to reach to the destination. For example, before starting the journey the human passenger allots certain amount of time to reach the destination and any deviation from that will be perceived as a \textit{gain} or \textit{loss} depending on whether they arrive earlier or later than their allotted time, respectively. 

Next, we define the utilities of possible outcomes and their probabilities. For any prospect or choice of route, the outcome is stochastic as it depends on various factors \cite{travel_info_impact}. Among these, we consider the effect of incoming traffic alerts from various traffic reporting platforms on the human passenger. Apart from the traffic alerts, we may also consider other stochastic factors that affect a human's decision to choose a route such as stochastic travel time, ease of navigation on a road, additional tariff etc. Here, we assume the utility of choosing a route to be a continuous random variable that depends on: (i) incoming traffic alerts, (ii) priority/veracity of these data, (iii) estimated travel time, and (iv) other fixed driving factors. Consequently, we represent the utility random variable as follows:
\begin{align}
    \mathbb{Z}=a_1S_1+\hdots+a_MS_M+k_1T+k_2,
\end{align}
where $S_1,\hdots,S_M\in \{0,1\}$ are incoming traffic alerts that follow Poisson Distribution; the weights $a_1,\hdots, a_M$ represents the decision maker's trust in or potential to act to each of these alerts; $T$ is stochastic travel time with $k_1$ as a weight; and $k_2$ captures the fixed driving convenience factors such as absence of toll, not leading to congested neighborhood detours, familiarity with the route, general road condition knowledge among many other factors. These weights are chosen to make $\mathbb{Z}$ dimension-less and to normalize each terms. Continuous travel time random variable $T$ is assumed to have a truncated normal distribution.

Thereafter, a logit model is  used to to predict the probability of an outcome. This enables us to capture the probabilistic nature of human decision making \cite{nilssonCPT}.  We define the utility random variable for the human passenger in $k$-th vehicle to be $\mathbb{Z}^k$ and the corresponding CPT (subjective) utility value to be $\mathbb{U}_j^k$. Hence, the probability of human passenger in $k$-th vehicle to choose Route $j$ is given by the logit model as follows \cite{nilssonCPT}:
\begin{align}
    p_j^k = \frac{e^{\phi  \mathbb{U}^k_j }}{e^{\phi  \mathbb{U}^k_1}+e^{\phi  \mathbb{U}^k_2}},\, \forall j \in \{1,2\},
\end{align}
where parameter $\phi>0$ is the sensitivity parameter which determines how the decision making is sensitive to individual utility. For example, with $\phi=0$, the choice is random and is unaffected by utility of either choice. With increasing $\phi$, the probability is affected increasingly more by the difference in utility of the choices. This can be easily seen by the alternate form of this logit model, for say Route 1: $p_1^k = {1}/(1+e^{\phi ( \mathbb{U}^k_2-\mathbb{U}^k_1)})$ where even small difference $(\mathbb{U}^k_2-\mathbb{U}^k_1)$ will significantly affect $p_1^k$ when $\phi$ is large.

Now, we compute $\mathcal{M}_j$, the number of human passengers who choose Route $j$ by aggregating the maximum probability of choosing that route for individual vehicles:
\begin{align}\label{Mj}
\mathcal{M}_j :=\left|\left\{m: \argmax\limits_{m\in \{1, 2\}} p_m^k=j, \,\forall k\in \{1, \hdots, \mathcal{M}\}\right\} \right|.    
\end{align}
We note here that the probability $p^k_j$ is dependent on the traffic alerts $S_1,\hdots,S_M$, weighting parameters $a_1,\hdots, a_M$, travel time $T$, scaling parameter $k_1$, driving convenience factors $k_2$ as well as parameters of human choice behavior given by $R, \beta^+, \beta^-, \lambda, \gamma$ and $\phi$. Defining a vector containing \textit{social} signals and parameters as  
\begin{align}\nonumber
    \sigma = [S_1,\hdots,S_M,a_1,\hdots, a_M,T, k_1, k_2, R, \beta^+, \beta^-, \lambda, \gamma, \phi ],
\end{align}
we can write $\mathcal{M}_j$ to be a function of $\sigma$, that is $\mathcal{M}_j(\sigma)$.

\begin{remm}
The parameters of the human behavior model can be estimated by collecting data through behavioral experiments or surveys \cite{zhang_frnd_cumulative,wang2018CPT_survey}. In order to capture the parameters accurately, the participant sample should also be varied  in terms of age, gender, race, financial standing and technological proficiency. These data collection strategies would record human route choice outcome under various scenarios of traffic alerts and other traffic information. 
\end{remm}

\subsection{CACC as a Differential Game}
In this problem, we assume that the vehicles in the traffic is Connected Adaptive Cruise Control (CACC) enabled. In this setting, the velocity of the vehicles are calculated with the information of other vehicles by solving an optimization problem \cite{wang2014noncoop}. Denoting the position and velocity of $k$-th vehicle opting Route $j$ by ${x}^{k}_j$ and ${v}^k_j$, respectively, the vehicle dynamics equation is given by  
\begin{align}\label{micro}
    \Dot{x}^{k}_j(\zeta) =v^k_j(\zeta),\, x^k_j(t) =  x^k_j,
\end{align}
for $ j\in\{1,2\}$ and $k \in \{1,\hdots,\mathcal{M}_j(\sigma)\}$ where $\mathcal{M}_j(\sigma)$ represents the number of vehicle choosing Route $j$. 
Additionally, for the $k$-th vehicle the driving overhead for choosing Route $j$ can be represented by the following functional:
\begin{align}\nonumber
    \mathcal{H}_j^k (v^k_j, {v}^{\sim k}_j) = &\int_{t_0}^t F^{\mathcal{M}_j}(v^k_j(\zeta),x^k_j(\zeta),\textbf{x}^{\sim k}_j(\zeta), \textbf{x}_{\sim j}(\zeta))\, d\zeta + R(x^k_j(t_0)),
\end{align}
where $\textbf{x}^{\sim k}_j(t) = \{x^1_j(t), \hdots, x^{k-1}_j(t), x^{k+1}_j(t), \hdots, x^{\mathcal{M}_j}_j(t)\}$ represents the positions of all vehicles except $k$-th vehicle that are opting for Route $j$, $\textbf{x}_{\sim j}$ represents positions of all vehicles opting for routes other than $j$, $v^k_j$ is the velocity of $k$-th vehicle opting for Route $j$ and ${v}^{\sim k}_j(t)$ is velocity of all vehicles other the $k$-th vehicle opting for Route $j$. Here, $F^{\mathcal{M}_j}$ is a cost functional whose form is identical for all vehicles along Route $j$ and $R(x^k_j(t_0))$ is the \textit{starting penalty} along Route $j$ which depends on the initial position of the $k$-th vehicle (Route $j$) at time $t=t_0$ (refer to \eqref{micro}). We assume that the cost functional $F^{\mathcal{M}_j}$ is strictly convex with respect to the variable $v^k_j(t)$ for the existence of solution to the Hamilton-Jacobi-Bellman (HJB) equation \cite{kirk2004}. 

In this setting, we assume that all the vehicles have same free-flow velocity $v_{max}$. Hence, the optimal velocity $u^k_j \in (0,v_{max}]$ for the $k$-th vehicle minimizes the driving cost among all other vehicles such that $ \mathcal{H}_j^k (u^k_j, {v}^{\sim k}_j) \leqslant  \mathcal{H}_j^k (v^k_j, {v}^{\sim k}_j)$. Consequently, every vehicle solves this optimal solution simultaneously to form a non-cooperative multi-vehicle differential games.

\subsection{Limiting Mean Field Game}

Generally speaking, solving the aforementioned non-cooperative multi-vehicle game is exceedingly hard with increased number of vehicles. This motivates the introduction of Mean Field Games (MFG) which is a non-cooperative game with arbitrarily large number of players \cite{cardaliaguet2010notes}. Unlike finite differential games where each player interacts with every other player, in MFG individual interactions are \textit{smoothed out} in the sense that coupling between players is only through the interaction with the average behavior or \textit{mean field}. In MFG limit of our differential games among vehicles, a global traffic behavior emerges from the collective interactions of vehicles, as derived in \cite{Di2019game}. 

Following \cite{Di2019game}, we use the concept of traffic density in order to smoothen the position information of the vehicles. This density function can be constructed using Kernel Density Estimation (KDE). Here, we use the Parzen-Rosenblatt window method \cite{parzen1962estimation,Rosenblatt1956} that \textit{smooths out} the position information over window length $a$ to produce local density information. First, the position information is captured using a Dirac comb function $C(x) = \frac{1}{N}\sum_{k=1}^N\delta(x-x_k)$. Subsequently, $C(x)$ is smoothed to the density function 
\begin{align}
    \rho^\mathcal{M} =\int_{\mathbb{R}}\Phi_a(x-y)C(y)dy= \frac{1}{\mathcal{M}}\sum_{k=1}^\mathcal{M} \Phi_a(x-x_j),
\end{align}
 where the Gaussian Smoothing kernel $\Phi_a (x)$ is given by 
\begin{align}
    \Phi_a (x) = \frac{1}{\sqrt{2\pi }a}\exp\left(-\frac{x^2}{2a^2}\right). \label{gau}
\end{align}
Accordingly, the smoothed total density function in our problem is derived as
\begin{align}\nonumber
    \rho^{\mathcal{M}}(x_1,x_2,t)=&\frac{1}{\mathcal{M}_1+\mathcal{M}_2}\Bigg[\sum\limits_{k=1}^{\mathcal{M}_1}\Phi_a(x_1-x_1^k(t))+\sum\limits_{k=1}^{\mathcal{M}_2}\Phi_a(x_2-x_2^k(t))\Bigg]\\\label{rho_alpha}
    =&\frac{\mathcal{M}_1}{\mathcal{M}_1+\mathcal{M}_2}\rho^{\mathcal{M}_1}(x_1,t) + \frac{\mathcal{M}_2}{\mathcal{M}_1+\mathcal{M}_2}\rho^{\mathcal{M}_2}(x_2,t),
\end{align}
where $\rho^{\mathcal{M}_j}(x_j,t)=\frac{1}{\mathcal{M}_j}\sum_{k=1}^{\mathcal{M}_j}\Phi_a(x_j-x_1^k)$ and $\mathcal{M}=\mathcal{M}_1+\mathcal{M}_2$ is the total number of vehicles. The global cost functional can now be expressed in terms of smooth density information instead of discrete positions of individual vehicles as follows \cite{cardaliaguet2010notes}:
\begin{align}
    &F^{\mathcal{M}_j}(v^k_j(t),x^k_j(t),\textbf{x}^{\sim k}_j(t), \textbf{x}_{\sim j}(t)):=F(v_j(t),\rho^{\mathcal{M}_j}(x_j(t),t),\rho^{\mathcal{M}_{\sim j}}(x_{\sim j}(t),t),\alpha_j),
\end{align}
where $\alpha_j= \frac{\mathcal{M}_j}{\mathcal{M}_1+\mathcal{M}_2}$. Here $\rho^{\mathcal{M}_{\sim j}}$ represents the density of the vehicles not choosing Route $j$. Notably, we observe here that  the driving cost functional is affected by  
parameter $\alpha_j$ which in turn is dependent on the output of human route choice model as shown in \eqref{Mj}.

Next, we obtain the MFG limit by making the number vehicles $\mathcal{M}_1$ and $\mathcal{M}_2$ to be arbitrarily large such that $\frac{\mathcal{M}_1}{\mathcal{M}_2}\to \kappa $ where $0<\kappa<\infty$.   This implies that limit to infinity of both kinds of vehicles must be of the same order. We also make the smoothing parameter $a$ in \eqref{gau} arbitrarily small such that $a/\mathcal{M}_j \to 0, \forall j$. This implies that on a finite road with increasing number of vehicles and shrinking \textit{window} of density contribution from each vehicle, the local density will eventually describe the global density. 

As $\mathcal{M}_j \to \infty$, $\rho^{\mathcal{M}_j}(x_j,t)\to\rho_j(x_j,t)$ which describes the density of the traffic when \textit{only} vehicles in Route $j$ was travelling on the road of interest. On the other hand, as both $\mathcal{M}_1,\mathcal{M}_2 \to \infty$ the \textit{effective} density of the traffic due to all classes of vehicles is given by $\rho(x,t)$. Then  \eqref{rho_alpha} yields
\begin{align}\label{rho}
    \rho(x,t)=\alpha \rho_1(x,t)+(1-\alpha)\rho_2(x,t),
\end{align}
where $\alpha = \frac{\kappa}{\kappa+1}$ and $\rho_j(x,t)=\rho(x_j(t),t)$. From the assumptions of MFG limit, we can easily derive the bounds of $\alpha$ to be $0<\alpha<1$.

\subsection{Socio-Technical Model}

In this subsection, we derive the socio-technical macroscopic traffic model in two steps. Note that there are two classes of vehicles each corresponding to a particular route choice. First, we combine the theoretical tools discussed in Sections III.A, III.B, and III.C to develop the dynamics of class-specific traffic state which is dependent on human route choice. Second, we derive the continuity equation for traffic flow model for each class of vehicles.  

Towards the first step, we consider the MFG setting where the position dynamics of the class of vehicles choosing Route $j$ be represented by 
\begin{align}
    \Dot{x_j}(\zeta)=v_j(\zeta), \quad x_j(t_0)=x_0, \quad \zeta \in [t_0,t].\label{eq11}
\end{align}
Then, we define the optimal cost functional $\mathcal{H}_j(x,t)$ for the class of  vehicles along Route $j$ to reach position $x$ at time $t$. 
  \begin{align}
    \mathcal{H}_j(x,t)&=R(x_j(t_0)) + \!\!\min\limits_{\substack{v_j(\tau)\\t_0\leqslant\tau \leqslant t}}\!\!\int_{t_0}^t\!\! F\, d\zeta, \label{cost_function_main}
\end{align}
where $F=F(v_j(\zeta),\rho_j(x,\zeta),\rho_{\sim j}(x,\zeta),\alpha)$, $v_j(\zeta)$ is generated from \eqref{eq11}, and $x_j(t)=x, \, j\in \{1,2\}$.

 From the fundamentals of dynamic programming, we know that the optimal cost functional $\mathcal{H}_j$ satisfies the HJB equation given as \cite{kirk2004}:
\begin{align}
  \frac{\partial \mathcal{H}_j(x,t)}{\partial t} =& \min\limits_{v_j(t)}\Bigg\{ F(v_j(t),\rho_j(x,t),\rho_{\sim j}(x,t),\alpha) - v_j(t) \frac{\partial \mathcal{H}_j}{\partial x}\Bigg\}. \label{hjb_0}
\end{align}

Thereafter, let us introduce the Legendre-Fenchel transform $F^\ast:(I^\ast\times \mathbb{R^+}\times \mathbb{R^+},\mathbb{R^+})\to \mathbb{R}$ where $I^\ast=\{p\in\mathbb{R}: \min\limits_{x\in\mathbb{R}}\{F(x, \rho_1,\rho_2,\alpha)-px\} <\infty\}$ and $F^\ast\left(p, \rho_1,\rho_2,\alpha\right):= \min\limits_{x \in \mathbb{R}} \left\{ F(x,\rho_1,\rho_2,\alpha)- x p\right\}.$ This transform is well-defined for a convex function $F$.  Furthermore, let us define 
$R(x_j(t_0)):=R_j(x_0)$. Also, define the optimal velocity solution for \eqref{hjb_0} to be $u_j(x,t)$. Then \eqref{hjb_0} can be equivalently written as a first order PDE 
\begin{align}\label{hjb}
    &\frac{\partial \mathcal{H}_j(x,t)}{\partial t} = F^\ast\left(\frac{\partial \mathcal{H}_j(x,t)}{\partial x}, \rho_j(x,t),\rho_{\sim j}(x,t),\alpha\right) ,
\end{align}
and the optimal solution can be written as 
\begin{align}
    u_j= F^\ast_w\left(w , \rho_j,\rho_{\sim j},\alpha\right)\Big\rvert_{w = \frac{\partial \mathcal{H}_j}{\partial x}}. \label{optimal}
\end{align}
The initial condition for \eqref{hjb} is given by $\mathcal{H}_j(x, t_0) = R_j(x_0)$.

Next, towards the second step, we derive the continuum equation which provides the dynamical equation for the density of vehicles. This is derived from the set of conservation laws for a class of vehicles choosing a specific route and is given below \cite{fan_multiclass}:
\begin{align}\label{cont_1}
    \frac{\partial \rho_j}{\partial t}+\frac{\partial (\rho_j u_j)}{\partial x}=0,\, j\in\{1,2\}.
\end{align}

Now, we finally present the multi-class traffic model with a human choice attribute.
The conservation of vehicles is obtained from \eqref{cont_1} while the dynamics of the driving cost variable is obtained  
from \eqref{hjb_0} and \eqref{optimal}.
The final model equations reads
\begin{align}\label{rho_dyn}
    &\frac{\partial \rho_j}{\partial t}+\frac{\partial (\rho_j u_j)}{\partial x}=0, \\ \label{H_dyn}
    &\frac{\partial \mathcal{H}_j}{\partial t}+u_j\frac{\partial \mathcal{H}_j}{\partial x}=F(u_j,\rho,\rho_{\sim j},\alpha),\\
    &u_j = \mathcal{I}\left(\frac{\partial \mathcal{H}_j}{\partial x}, \rho_j,\rho_{\sim j},\alpha\right), \label{FD_u}
\end{align}
where $\mathcal{I} = F^\ast_w(w,.)$ represents the velocity-density relation (Fundamental Diagram) which depends on the densities of the both class of vehicles. Note that $\mathcal{I}$ is parameterized by the driving cost function $\frac{\partial \mathcal{H}_j}{\partial x}$ and human route choice parameter $\alpha$. We also note from \eqref{cost_function_main} that $\mathcal{H}_j$ depends on human route choice parameter $\alpha$. This in turn implies that the Fundamental Diagram is dependent on $\alpha$ implicitly through $\mathcal{H}_j$ as well.

Defining a state vector
\begin{align}\label{eta}
     \eta(x,t) = [\rho_1(x,t),u_1(x,t),\rho_2(x,t), u_2(x,t)]^T,
\end{align}
we can write the boundary conditions for this system as
\begin{align}\label{BC}
    \eta(B^-,t) = G_B\eta(A^+,t), \, \eta(C^-,t) = G_C \eta(B^+,t), \forall t\in \mathbb{R}^+,
\end{align}
where $G_B$ and $G_C$ are $4\times 4$ matrices. Moreover, the Rankine-Hugoniot Condition \cite{benedetto_book} provides the connection formula for conservation of fluxes before and after the juncture point B and is given by:
\begin{align}\label{RH1}
    &\alpha u_1(B^-,t)\rho_1(B^-,t)+(1-\alpha)u_2(B^-,t)\rho_2(B^-,t) \\
    &= u_1(B^+,t)\rho_1((B^+,t)+u_2(B^+,t)\rho_2(B^+,t).
\end{align}
The initial condition for the model is
\begin{align}\label{IC}
    \eta(x,0)=\eta_0.
\end{align}
Finally, the socio-technical model for human route choice is given by dynamical equations \eqref{rho_dyn}-\eqref{FD_u}, boundary condtions \eqref{BC}, initial condition \eqref{IC} and connection formula for the road juncture at B \eqref{RH1}.

\section{Mathematical Characteristics of the Socio-Technical Model}
Depending on the driving objective of the CACC, we can choose various cost functionals $F(.)$ in \eqref{H_dyn}. To capture comparable impact of the terms in the cost functional, we assume that all variables are normalized.

\begin{align}
    &F(u_j,\rho_1,\rho_2,\alpha) =\mathfrak{L}(x) F_{j}+ (1-\mathfrak{L}(x))G_{j},\label{functional}\\
    &F_{1} = \frac{u_1^2}{2}-u_1 + \alpha u_1 \rho_1 , \label{functional1a}\\
    &F_{2} = \frac{u_2^2}{2}-u_2 +(1-\alpha) u_2 \rho_{2},  \label{functional1b}\\
    &G _j = \frac{u_j^2}{2}-u_j + u_j \rho_{ j},   \forall j\in \{1,2\}.\label{functional2}
\end{align}
The cost functional $F_{j}$ corresponds to the vehicles driving on main road (between A and B), and the cost functional $G_{j}$ corresponds to the vehicles driving on the Route $j$ (between C and D). The function $\mathfrak{L}=1$ when the cars are in AB while  $\mathfrak{L}=0$ when they are on Route $j \in \{1,2\}$. For functional $F_{j}$ in \eqref{functional1a}-\eqref{functional1b}, the first and second terms represent the kinetic energy and driving efficiency of the class of vehicles choosing Route $j$ whereas the last term represents the driving safety of the vehicles of class $j$ in response to the density of the vehicles of same class on the road.  Similarly, for the functional $G_{j}$ in \eqref{functional2}, terms include kinetic energy, efficiency and the safety of vehicles on Route $j$. 

Next, we analyse the model for $\mathfrak{L}=1$. The analysis for $\mathfrak{L}=0$ can be done in a similar manner. Consequently, we analyze the system in $x\in [A,B]$ and argue that the model can be similarly analyzed in the domain $x \in [C,D]$. Moreover, for simplicity we have assumed  the \textit{starting penalty} \mbox{$R(x_j(t_0))=R_j(x_0)$} in \eqref{cost_function_main} to be zero. We analyze a few salient characteristics of the socio-technical traffic model: (i) we examine the Fundamental Diagram of the model, (ii) investigate the criteria for hyperbolicity of the model, and (iii) analyze the linear stability based on linearized version of the model, around operating points of the CACC-enabled vehicular traffic.
\subsection{Fundamental Diagram}
Due to the specific form of functional chosen in \eqref{functional1a}-\eqref{functional1b}, we can calculate the velocity-density relation from \eqref{optimal} to be:
\begin{align}\label{vel_specific}
    u_1 &= \mathcal{I}\left(\frac{\partial \mathcal{H}_1}{\partial x}, \rho_1,\alpha\right)=1+ \frac{\partial\mathcal{H}_1}{\partial x} -\alpha\rho_1 , \\
       u_2 &= \mathcal{I}\left(\frac{\partial \mathcal{H}_2}{\partial x}, \rho_2,\alpha\right)=1+ \frac{\partial\mathcal{H}_2}{\partial x} -(1-\alpha)\rho_2 , \label{vel_specific11}
\end{align}
Next, let us define the effective densities for the two classes of vehicles $d_1 = \alpha \rho_1$ and $d_2 = (1-\alpha)\rho_2$. Using \eqref{rho}, the total effective density is given by $d_1+d_2=\rho$. Then \eqref{vel_specific} and \eqref{vel_specific11} can be re-written as 
\begin{align}\label{vel_eff}
     u_j &= \mathcal{I}\left(\frac{\partial \mathcal{H}_j}{\partial x}, d_j\right)=1+ \frac{\partial\mathcal{H}_j}{\partial x} -d_j.
\end{align}
This indicates that for the chosen cost functional, velocities of both classes of vehicles depend the optimal driving cost $\mathcal{H}_j$ whereas the latter depends on $\alpha$ (see \eqref{cost_function_main}). This implies that the velocity-effective density relation also implicitly depends on $\alpha$. Hence, $\frac{\partial\mathcal{H}_j}{\partial x}$ is termed as a human route choice attribute for the Fundamental Diagram.

Now, for equilibrium flow in this linear model, we have $\frac{\partial\mathcal{H}_j}{\partial x}=0$ \cite{Di2}. This leads to $u_j = 1 -d_j$ in \eqref{vel_eff}. This implies that under equilibrium condition, both classes of vehicles follow a Greenshields Fundamental Diagram with slope $-1$ and maximum velocity $1$. Now, under non-equilibrium condition, the maximum velocity depends on the socio-technical parameter of each class of vehicles as $u_j\big|_{max} = 1+ \frac{\partial\mathcal{H}_j}{\partial x}$ while the slope of Fundamental Diagram still being $-1$. The effect of the human route choice attribute $\frac{\partial\mathcal{H}_j}{\partial x}$ has been shown in the Fundamental Diagram in Fig. \ref{fig:FD}.

\begin{figure}[t]
    \centering
    \includegraphics[trim = 0mm 0mm 0mm 0mm, clip,  scale=1.0, width=0.3\linewidth]{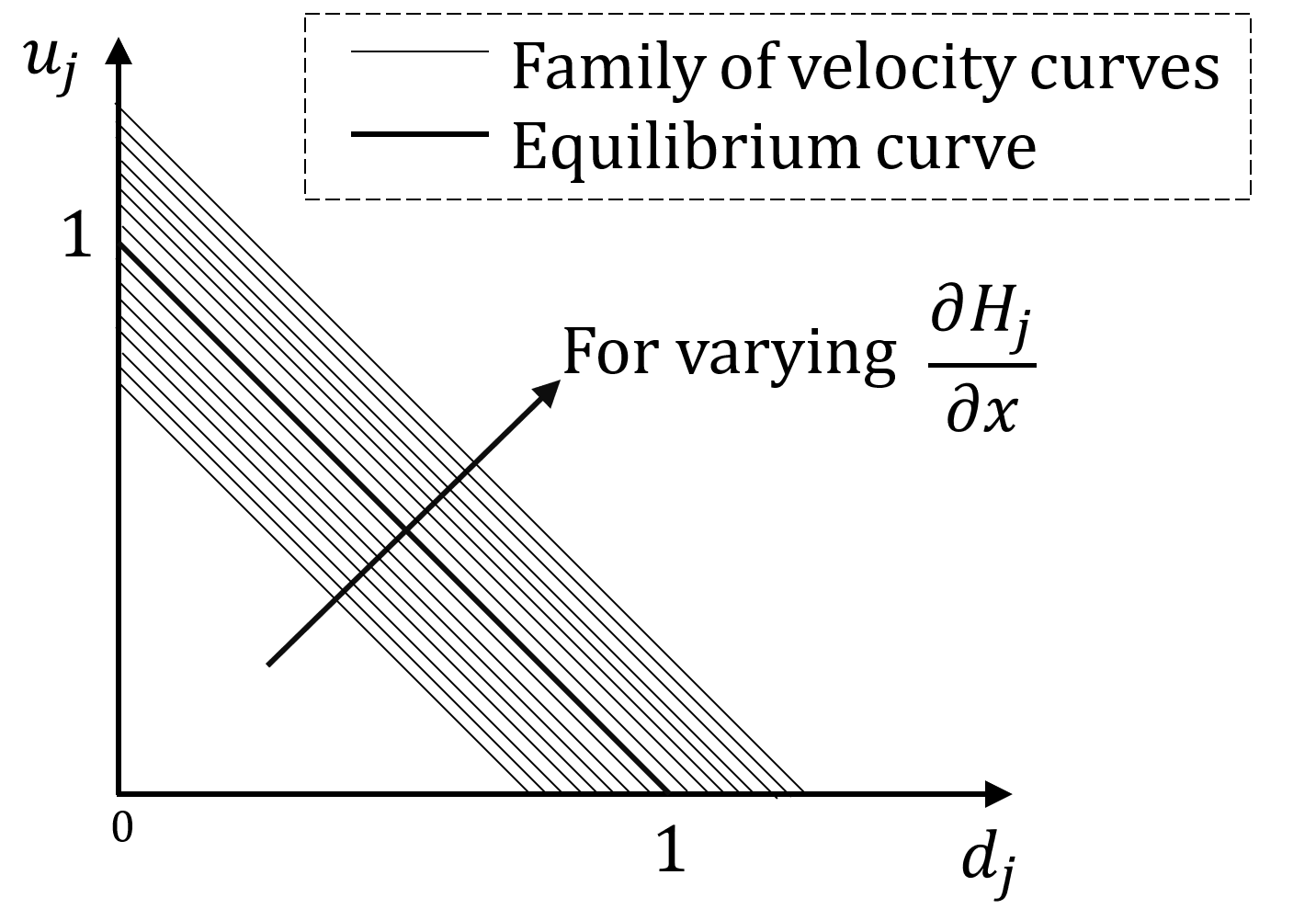}
    \caption{Velocity vs effective density Fundamental Diagram showing variation due to human route choice attribute $\frac{\partial \mathcal{H}_j}{\partial x}$.}
    \label{fig:FD}
\end{figure}

\subsection{Linearized Traffic Model}

We can further simplify the system by eliminating $\mathcal{H}_j$ from the equations \eqref{vel_specific}-\eqref{vel_specific11} using \eqref{rho_dyn}-\eqref{H_dyn} to obtain:
\begin{align}\label{nonlin1}
    &\frac{\partial u_1}{\partial t} + u_1 \frac{\partial u_1}{\partial x} -\alpha\frac{\partial (\rho_1 u_1)}{\partial x}  =0,\\\label{nonlin2}
    &\frac{\partial u_2}{\partial t} + u_2 \frac{\partial u_1}{\partial x} -(1-\alpha)\frac{\partial (\rho_2 u_2)}{\partial x} =0.
\end{align}
The above equations along with the continuity equations in \eqref{cont_1} form the socio-technical model the traffic in this case study.  This system of equations can be alternatively written in the form of a conservation law as below:
\begin{align}\label{nonlin_eta}
    &\eta_t + \mathcal{Q}(\eta)_x=0,
\end{align}
where $\eta$, defined in \eqref{eta}, denotes the traffic state. The flux $\mathcal{Q}$ of $\eta$ is given by 
\begin{align}
    \mathcal{Q}([\rho_1,u_1,\rho_2, u_2]^T) = \begin{bmatrix}
    \rho_1u_1\\
    \frac{u_1^2}{2} -\alpha\rho_1u_1 \\
    \rho_2u_2\\
    \frac{u_2^2}{2} -(1-\alpha)\rho_2u_2
    \end{bmatrix}.
\end{align}

Since our traffic is assumed to be under CACC control strategy, we can linearize the system around its operating points $\eta^\ast = [ \rho^\ast_1,u^\ast_1,\rho^\ast_2,  u^\ast_2]^T$ for the two classes of vehicles and subsequently analyze its properties from the linearized model. To linearize, we take a Taylor series expansion of $\mathcal{Q}$ around operating point $\eta^\ast$ in \eqref{nonlin_eta} and  neglect higher order terms to obtain:
\begin{align}\label{eta_model}
    &\eta_t + ({\partial \mathcal{Q}}/{\partial \eta})\big|_{\eta=\eta^\ast} \eta_x = 0,
\end{align}
Here $J = ({\partial \mathcal{Q}}/{\partial \eta})|_{\eta=\eta^\ast}$ is the Jacobian of the system and is given by
\begin{align}\label{Jacobian}
    &J = \frac{\partial \mathcal{Q}}{\partial \eta}\Bigg|_{\eta=\eta^\ast}  = \begin{bmatrix}
    A&0\\0&B
    \end{bmatrix}, \\\label{A}
    &A = \begin{bmatrix}
    u_1^\ast & \rho_1^\ast  \\
   -\alpha u_1^\ast  & u_1^\ast -\alpha\rho_1^\ast  
    \end{bmatrix},
   \\\label{D}
    & B = \begin{bmatrix}
    u_2^\ast  & \rho_2^\ast \\
    -(1-\alpha) u_2^\ast  & u_2^\ast -(1-\alpha)\rho_2^\ast 
    \end{bmatrix}.
\end{align}

\subsection{Hyperbolicity of the Traffic Model}
In a qualitative sense, hyperbolicity of a Partial Differential Equation (PDE) system reflects a wave-like nature of its solution. This implies that disturbances to the system propagate at finite speeds along the characteristics of the equations. For a PDE model to be considered as a traffic model, it is imperative to prove that such system is hyperbolic in nature. In this section, we prove that the socio-technical model considered in our case study is strictly hyperbolic. 
\begin{deff}[Strict hyperbolicity \cite{evans_PDE}]
A PDE system is strictly hyperbolic if and only if all the eigenvalues of its Jacobian are real and distinct. 
\end{deff}

\begin{Lemm}[Strict hyperbolicity of the socio-technical traffic model]
Consider the multi-class socio-technical traffic model given in \eqref{eta_model}-\eqref{D}. For a given non-zero operating density and velocity for each route $j\in \{1,2\}$ to be $(\rho_j^\ast,u_j^\ast)$, the traffic model is strictly hyperbolic (with negative definite Jacobian $J$) if and only if the following conditions are satisfied:
\begin{align}\label{hyper_cond}
    \frac{4u_1^\ast}{\rho^\ast_1}<\alpha,\quad \frac{4u_2^\ast}{\rho^\ast_2}<(1-\alpha).
\end{align}
\end{Lemm}
\begin{proof}
Since $B^2=0$, the characteristic equation $\mathcal{P}(\lambda)$ of the Jacobian $\frac{\partial \mathcal{Q}}{\partial \eta}$ simplifies to
\begin{align}
\mathcal{P}(\lambda):=\det(\lambda I-A)\det(\lambda I-B).
\end{align}
The four eigenvalues can be computed from the roots of $\mathcal{P}(\lambda)=0$ as
\begin{align}
    &\lambda_{1,2} = \frac{1}{2}\left[u_1^\ast-\alpha\rho_1^\ast \pm \alpha\rho_1^\ast\sqrt{1-\frac{4u_1^\ast}{\alpha\rho_1^\ast}}\, \right],\\
    &\lambda_{3,4} = \frac{1}{2}\left[u_2^\ast-(1-\alpha)\rho_2^\ast \pm (1-\alpha)\rho_2^\ast\sqrt{1-\frac{4u_2^\ast}{(1-\alpha)\rho_2^\ast}}\, \right].
\end{align}
The discriminants here are given by 
$
    \Delta_1 = 1-\frac{4u_1^\ast}{\alpha\rho_1^\ast}, 
    \Delta_2 = 1-\frac{4u_2^\ast}{(1-\alpha)\rho_2^\ast}.
$
It is evident that if  the discriminants $\Delta_1,\Delta_2$ will be positive, the roots will be real and distinct. Thus under conditions of \eqref{hyper_cond} the eigenvalues of the system are real and distinct and  \eqref{eta_model} is strictly hyperbolic.

Lastly, since operating density, $\rho_j^\ast$, velocity $u_j^\ast$ and $\alpha$ are all positive, we note here that $\Delta_1,\Delta_2<1$. Using \eqref{hyper_cond} with the upper bound on $\Delta_j$, one can easily show that all the eigenvalues $\lambda_m, m\in\{1,2,3,4\}$ are \textit{negative} and $J$ is negative definite. 
\end{proof}

\subsection{Linear Stability Analysis}
In this subsection, we analyze the linear stability of the traffic model \eqref{eta_model}.

\begin{deff}[Exponential Stability \cite{expo_stability}] The linear hyperbolic system given by \eqref{eta_model}-\eqref{D} along with its boundary \eqref{BC} and initial conditions \eqref{IC} is exponentially stable around operating point $\eta^\ast$ in $x\in[A,B]$ if there exists an $\epsilon>0$ and  $0<M<\infty$ such that for every initial condition $\eta_0 \in L^2([A,B]; \mathbb{R}^4)$, the solution of the boundary value problem satisfies
\begin{align}\label{stability_defn}
    \|\eta(.,t)-\eta^\ast\|_{L^2([A,B]; \mathbb{R}^4)}\leqslant Me^{-\epsilon t }\|\eta_0-\eta^\ast\|_{L^2([A,B]; \mathbb{R}^4)}, 
\end{align}
for all $t\in[0,\infty).$
\end{deff}

\begin{remm}
Exponential stability of the traffic system implies that the traffic state $\eta$ will converge exponentially to operating condition $\eta^\ast$ as $t\to \infty$ given bounded energy error in the initial condition of the system. This indicate that the density and velocity of both the classes of vehicles converge to operating condition as $t\to \infty$. In other words, $\rho_j\to \rho^*_j$ and $u_j\to u_j^\ast$ for $j\in\{1,2\}$ as $t\to \infty$. Additionally, the rate of decay $\epsilon$ depends on system characteristics, in particular, on the Jacobian $J$ for our problem. Since, the Jacobian is a function of human route choice attribute $\alpha$, the convergence also depends on the same.
\end{remm}

\begin{thmm}[Exponential Stability Condition]
Consider the Socio-technical hyperbolic traffic model in \eqref{eta_model}-\eqref{D} along with its boundary \eqref{BC} and initial conditions \eqref{IC}  that satisfies condition \eqref{hyper_cond}. This system is exponentially stable in the sense of \eqref{stability_defn} around operating point $\eta^*$ if there exist a $\mu>0$ such that the matrix $\mathcal{J}$, given by
\begin{align}
    \mathcal{J}=J-G_B^TJG_Be^{\mu (A-B)},
\end{align}
is positive definite.
\end{thmm}
\begin{proof}
We define the deviation in the traffic states from the operating point $\eta^\ast$ as a new vector
\begin{align}
    E(x,t) = \eta(x,t) - \eta^\ast.
\end{align}
Thereafter, we choose the Lyapunov functional candidate 
\begin{align}
     V(t) =& \frac{1}{2}\int_A^Be^{-\mu \zeta}E^T(\zeta,t)E(\zeta,t)\, d\zeta, \label{Lyap}
\end{align}
where $0<\mu<1$. We can re-write \eqref{Lyap} as
\begin{align}\label{V_defn}
    &V(t)=\frac{1}{2}\|E(.,t)e^{-\frac{\mu}{2} (.)}\|^2_{L^2([A,B])},
\end{align}   
where $\|f(.,t)\|^2_{L^2([A,B])}=\int_A^Bf^2(\zeta,t)d\zeta$. The initial condition $V(0)$ is given by $V(0)=\frac{1}{2}\|E(.,0)e^{-\frac{\mu}{2} (.)}\|_{L^2([A,B])}$.

Taking derivative of $E(t)$ with respect to $t$ yields
\begin{align}\nonumber
    \dot{V}(t)& = \int_A^B e^{-2\mu \zeta}E_t^T(\zeta,t)E(\zeta,t)\, d\zeta\\\nonumber
    & = e^{-\mu A}E_\zeta^T(A,t)Je(A,t)-e^{-\mu B}E_\zeta^T(B,t)JE(B,t)\\
    &\hspace{1.3in}+\mu\int_A^B E^TJE e^{-\mu\zeta}\, d\zeta.
        \end{align}
From Lemma 1, we know that the Jacobian $J$ is negative definite which implies there exists $\lambda>0$ such that $E^TJE \leqslant -\lambda E^TE$. Additionally, we invoke the boundary condition \eqref{BC} to obtain 
    \begin{align}
    & \dot{V}(t) \leqslant -2\mu \lambda V(t) -E^T(A,t)\mathcal{J}E(A,t)e^{-\mu A},
\end{align}
where $\mathcal{J}=J-G_B^TJG_Be^{\mu (A-B)}$. If $\mathcal{J}$ is positive definite, we can write $\dot{V}(t)\leqslant -2\mu \lambda V(t) $, which implies $V(t)\leqslant e^{-2\mu \lambda t} V(0)$. Subsequently, we define 
\begin{align}\label{gamma}
    0<\gamma_1 = \min_{\zeta \in [A,B]}e^{-\frac{\mu \zeta}{2}}\leqslant\gamma_2 = \max_{\zeta \in [A,B]}e^{-\frac{\mu \zeta}{2}}<\infty.
\end{align}
Using these definitions in \eqref{V_defn}, initial condition expression $V(0)$, and \eqref{gamma}, we can write $\gamma^2_1(B-A)\|E(.,t)\|^2_{L^2([A,B])}\leqslant V(t)$ and $V(0)\leqslant \gamma^2_2(B-A)\|E(.,0)\|^2_{L^2([A,B])}$. Therefore we obtain
\begin{align}
   \|E(.,t)\|_{L^2([A,B])}\leqslant \left( \frac{\gamma_2}{\gamma_1}\right)e^{-\mu \lambda t}\|E(.,0)\|_{L^2([A,B])}.
\end{align}
For $0<M= \frac{\gamma_2}{\gamma_1}<\infty$ and $\epsilon = \mu \lambda>0$, we proved the socio-technical traffic system \eqref{eta_model}-\eqref{D} along with its boundary \eqref{BC} and initial conditions \eqref{IC}  is exponentially stable in the sense of \eqref{stability_defn}.
\end{proof}

\section{Simulation Results}

In this section, we perform simulation studies to illustrate the characteristics of the proposed model, as discussed in previous sections. The distributed plot of the traffic density and velocity for the two classes of the vehicles choosing a route are shown in Fig. \ref{fig:density} and Fig. \ref{fig:density2}. For this simulation, we have chosen $\alpha = 0.45$. The normalized operating density-velocity pair for the first class has been chosen as (0.85, 0.09,) while for second class it is chosen as (0.75,0.095). We note here that these operating values satisfy the conditions provided in Lemma 1. The eigenvalues of the Jacobian in this case are $-0.0549, -0.1476,   -0.0534$ and  $-0.1691$. This implies that the simulated socio-technical traffic model is strictly hyperbolic. Moreover, from both Fig. \ref{fig:density} and Fig. \ref{fig:density2}, we can observe that the system stabilizes to its normalized operating point after perturbation in initial conditions. Hence, the system is stable in the sense of \eqref{stability_defn}.

\begin{figure}[ht]
    \centering
    \includegraphics[trim = 0mm 0mm 0mm 0mm, clip, scale=1.0, width=0.5\linewidth]{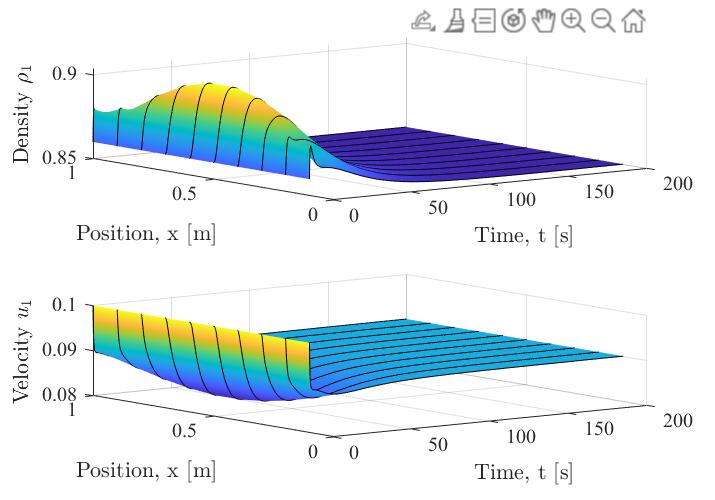}
    \caption{Plot of distributed density and velocity of Class 1 of traffic that chooses Route 1.}
    \label{fig:density}
\end{figure}

\begin{figure}[ht]
    \centering
    \includegraphics[trim = 0mm 0mm 0mm 0mm, clip, scale=1.0, width=0.5\linewidth]{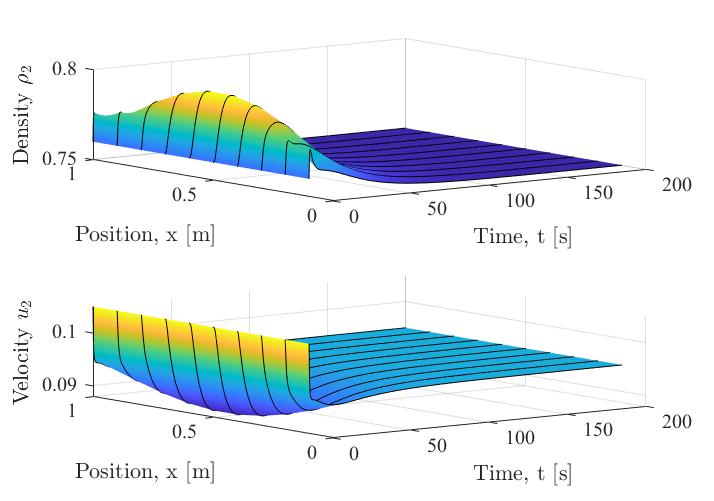}
    \caption{Plot of distributed density and velocity of Class 2 of traffic that chooses Route 2.}
    \label{fig:density2}
\end{figure}

Next, we investigate the impact of the human route choice attribute $\alpha$ on the socio-technical traffic model. Particularly, we explore the affect of $\alpha$ on the stability of the proposed model. We note that $\alpha$ increases as more and more human passengers choose Route 1. As expected, with increasing $\alpha$, the density of class 1 vehicles increases and their velocity decreases. This implies that perturbations to the system show larger overshoots in density and takes longer to converge to operating conditions. On the other hand, the contribution of vehicles in class 2 decreases which corresponds to lower density, higher speed and overall faster recovery from perturbations. These phenomena can be observed clearly from Fig. \ref{fig:stability}, where we plot the spatial norm of the perturbation of traffic density and velocities from operating points of each class. We also note that the simulated system were all stable, as can be seen from the decay of perturbations to zero, albeit at different rates for different $\alpha$.

\begin{figure}[ht]
    \centering
    \includegraphics[trim = 0mm 0mm 0mm 0mm, clip,  scale=1.0, width=0.5\linewidth]{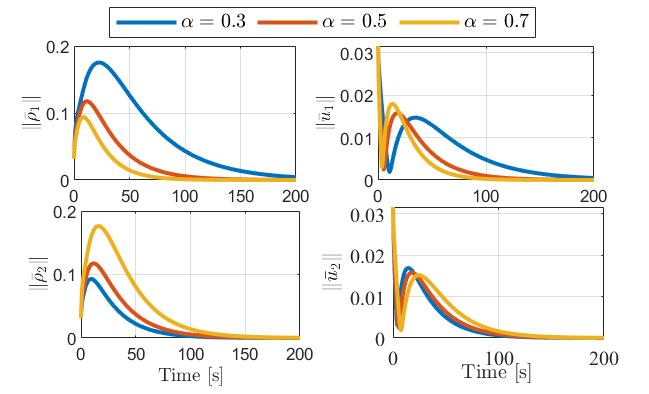}
    \caption{Convergence plots for traffic states' deviations for each class under different values of $\alpha$.}
    \label{fig:stability}
\end{figure}

\section{Conclusion}
In this paper, we proposed a socio-technical macroscopic traffic model for CACC enabled vehicles that captures the effect of human passengers' route choice behavior. Essentially, the human passengers' choice of routes leads to multi-class traffic system where each class of vehicles corresponds to each route. First, we have used CPT to model the influence of traffic alerts on human passengers' route choice behavior. Next, we have utilized non-cooperative differential game and MFG to obtain the macroscopic model of a multi-class traffic where the human route choice characterizes the Fundamental Diagram as well as the dynamics of traffic states. To validate the model characteristics, we perform a case study for a particular driving cost function and present simulation results for the same.

\bibliographystyle{acm} 
\bibliography{ref1}

\begin{thebibliography}{10}

\bibitem{Wait_TimeCPT}
{\sc Avineri, E., and Prashker, J.~N.}
\newblock Sensitivity to travel time variability: travelers’ learning
  perspective.
\newblock {\em Transportation Research Part C: Emerging Technologies 13}, 2
  (2005), 157--183.

\bibitem{travel_info_impact}
{\sc Ben-Elia, E., Di~Pace, R., Bifulco, G.~N., and Shiftan, Y.}
\newblock The impact of travel information’s accuracy on route-choice.
\newblock {\em Transportation Research Part C: Emerging Technologies 26\/}
  (2013), 146--159.

\bibitem{cardaliaguet2010notes}
{\sc Cardaliaguet, P.}
\newblock Notes on mean field games.
\newblock Tech. rep., Technical report, 2010.

\bibitem{HOV_lanechange_CPT}
{\sc Chow, J.~Y., Lee, G., and Yang, I.}
\newblock Genetic algorithm to estimate cumulative prospect theory parameters
  for selection of high-occupancy-vehicle lane.
\newblock {\em Transportation research record 2157}, 1 (2010), 71--77.

\bibitem{davis2004effect}
{\sc Davis, L.}
\newblock Effect of adaptive cruise control systems on traffic flow.
\newblock {\em Physical Review E 69}, 6 (2004), 066110.

\bibitem{delis2015macroscopic}
{\sc Delis, A.~I., Nikolos, I.~K., and Papageorgiou, M.}
\newblock Macroscopic traffic flow modeling with adaptive cruise control:
  Development and numerical solution.
\newblock {\em Computers \& Mathematics with Applications 70}, 8 (2015),
  1921--1947.

\bibitem{eagle2014reality}
{\sc Eagle, N., and Greene, K.}
\newblock {\em Reality mining: Using big data to engineer a better world}.
\newblock MIT Press, 2014.

\bibitem{evans_PDE}
{\sc Evans, L.}
\newblock {\em Partial Differential Equations}.
\newblock Graduate studies in mathematics. American Mathematical Society, 2010.

\bibitem{fan_multiclass}
{\sc Fan, S., and Work, D.~B.}
\newblock A heterogeneous multiclass traffic flow model with creeping.
\newblock {\em SIAM Journal on Applied Mathematics 75}, 2 (2015), 813--835.

\bibitem{fennema1997original}
{\sc Fennema, H., and Wakker, P.}
\newblock Original and cumulative prospect theory: A discussion of empirical
  differences.
\newblock {\em Journal of Behavioral Decision Making 10}, 1 (1997), 53--64.

\bibitem{chicken}
{\sc Fox, C., Camara, F., Markkula, G., Romano, R., Madigan, R., Merat, N.,
  et~al.}
\newblock When should the chicken cross the road?: Game theory for autonomous
  vehicle-human interactions.

\bibitem{gan_msg_sign}
{\sc Gan, H., and Ye, X.}
\newblock Whether to enter expressway or not? the impact of new variable
  message sign information.
\newblock {\em Journal of Advanced Transportation 49}, 2 (2015), 267--278.

\bibitem{gao_route}
{\sc Gao, S., Frejinger, E., and Ben-Akiva, M.}
\newblock Adaptive route choices in risky traffic networks: A prospect theory
  approach.
\newblock {\em Transportation research part C: emerging technologies 18}, 5
  (2010), 727--740.

\bibitem{benedetto_book}
{\sc Garavello, M., and Piccoli, B.}
\newblock {\em Traffic flow on networks}, vol.~1.
\newblock American institute of mathematical sciences Springfield, 2006.

\bibitem{gigerenzer_bounded_rational}
{\sc Gigerenzer, G., and Selten, R.}
\newblock {\em Bounded rationality: The adaptive toolbox}.
\newblock MIT press, 2002.

\bibitem{AnnaswamyCPT}
{\sc Guan, Y., Annaswamy, A.~M., and Tseng, H.~E.}
\newblock Cumulative prospect theory based dynamic pricing for shared mobility
  on demand services.
\newblock In {\em 2019 IEEE 58th Conference on Decision and Control (CDC)\/}
  (2019), IEEE, pp.~2239--2244.

\bibitem{CPT_stakelberg}
{\sc Han, Q., Dellaert, B.~G., Van~Raaij, W.~F., and Timmermans, H.~J.}
\newblock Integrating prospect theory and stackelberg games to model strategic
  dyad behavior of information providers and travelers: Theory and numerical
  simulations.
\newblock {\em Transportation research record 1926}, 1 (2005), 181--188.

\bibitem{Di2}
{\sc {Huang}, K., {Di}, X., {Du}, Q., and {Chen}, X.}
\newblock Stabilizing traffic via autonomous vehicles: A continuum mean field
  game approach.
\newblock In {\em 2019 IEEE Intelligent Transportation Systems Conference
  (ITSC)\/} (2019), pp.~3269--3274.

\bibitem{Di2019game}
{\sc Huang, K., Di, X., Du, Q., and Chen, X.}
\newblock A game-theoretic framework for autonomous vehicles velocity control:
  Bridging microscopic differential games and macroscopic mean field games.
\newblock {\em Discrete \& Continuous Dynamical Systems - B 25}, 12 (2020),
  4869–4903.

\bibitem{smartgrid_CPT}
{\sc Jhala, K., Natarajan, B., and Pahwa, A.}
\newblock Prospect theory-based active consumer behavior under variable
  electricity pricing.
\newblock {\em IEEE Transactions on Smart Grid 10}, 3 (2018), 2809--2819.

\bibitem{kirk2004}
{\sc Kirk, D.~E.}
\newblock {\em Optimal control theory: an introduction}.
\newblock Courier Corporation, 2004.

\bibitem{li2017game}
{\sc Li, N., Oyler, D.~W., Zhang, M., Yildiz, Y., Kolmanovsky, I., and Girard,
  A.~R.}
\newblock Game theoretic modeling of driver and vehicle interactions for
  verification and validation of autonomous vehicle control systems.
\newblock {\em IEEE Transactions on control systems technology 26}, 5 (2017),
  1782--1797.

\bibitem{milanes_CACC}
{\sc Milan{\'e}s, V., and Shladover, S.~E.}
\newblock Modeling cooperative and autonomous adaptive cruise control dynamic
  responses using experimental data.
\newblock {\em Transportation Research Part C: Emerging Technologies 48\/}
  (2014), 285--300.

\bibitem{ngoduy2013instability}
{\sc Ngoduy, D.}
\newblock Instability of cooperative adaptive cruise control traffic flow: A
  macroscopic approach.
\newblock {\em Communications in Nonlinear Science and Numerical Simulation
  18}, 10 (2013), 2838--2851.

\bibitem{nikolos_CACC_macroscopic}
{\sc Nikolos, I.~K., Delis, A.~I., and Papageorgiou, M.}
\newblock Macroscopic modelling and simulation of acc and cacc traffic.
\newblock In {\em 2015 IEEE 18th International Conference on Intelligent
  Transportation Systems\/} (2015), IEEE, pp.~2129--2134.

\bibitem{nilssonCPT}
{\sc Nilsson, H., Rieskamp, J., and Wagenmakers, E.-J.}
\newblock Hierarchical bayesian parameter estimation for cumulative prospect
  theory.
\newblock {\em Journal of Mathematical Psychology 55}, 1 (2011), 84--93.

\bibitem{parzen1962estimation}
{\sc Parzen, E.}
\newblock On estimation of a probability density function and mode.
\newblock {\em The annals of mathematical statistics 33}, 3 (1962), 1065--1076.

\bibitem{prelec}
{\sc Prelec, D.}
\newblock The probability weighting function.
\newblock {\em Econometrica 66}, 3 (1998), 497--527.

\bibitem{Rosenblatt1956}
{\sc Rosenblatt, M.}
\newblock Remarks on some nonparametric estimates of a density function.
\newblock {\em Annals of Mathematical Statistics 27\/} (1956), 832--837.

\bibitem{sadigh2016planning}
{\sc Sadigh, D., Sastry, S., Seshia, S.~A., and Dragan, A.~D.}
\newblock Planning for autonomous cars that leverage effects on human actions.
\newblock In {\em Robotics: Science and Systems\/} (2016), vol.~2, Ann Arbor,
  MI, USA.

\bibitem{talebpour2015modeling}
{\sc Talebpour, A., Mahmassani, H.~S., and Hamdar, S.~H.}
\newblock Modeling lane-changing behavior in a connected environment: A game
  theory approach.
\newblock {\em Transportation Research Part C: Emerging Technologies 59\/}
  (2015), 216--232.

\bibitem{CPT_main}
{\sc Tversky, A., and Kahneman, D.}
\newblock Advances in prospect theory: Cumulative representation of
  uncertainty.
\newblock {\em Journal of Risk and uncertainty 5}, 4 (1992), 297--323.

\bibitem{expo_stability}
{\sc Vamvoudakis, K., and Jagannathan, S.}
\newblock {\em Control of Complex Systems: Theory and Applications}.
\newblock Butterworth-Heinemann, 2016.

\bibitem{wang2014noncoop}
{\sc Wang, M., Daamen, W., Hoogendoorn, S.~P., and van Arem, B.}
\newblock Rolling horizon control framework for driver assistance systems. part
  i: Mathematical formulation and non-cooperative systems.
\newblock {\em Transportation research part C: emerging technologies 40\/}
  (2014), 271--289.

\bibitem{wang2014coop}
{\sc Wang, M., Daamen, W., Hoogendoorn, S.~P., and van Arem, B.}
\newblock Rolling horizon control framework for driver assistance systems. part
  ii: Cooperative sensing and cooperative control.
\newblock {\em Transportation research part C: emerging technologies 40\/}
  (2014), 290--311.

\bibitem{wang2013coop}
{\sc Wang, M., Treiber, M., Daamen, W., Hoogendoorn, S.~P., and van Arem, B.}
\newblock Modelling supported driving as an optimal control cycle: Framework
  and model characteristics.
\newblock {\em Transportation Research Part C: Emerging Technologies 36\/}
  (2013), 547--563.

\bibitem{wang2018CPT_survey}
{\sc Wang, S., and Zhao, J.}
\newblock How risk preferences influence the usage of autonomous vehicles.
\newblock Tech. rep., 2018.

\bibitem{whitworth2013social}
{\sc Whitworth, B., and Ahmad, A.}
\newblock {\em The social design of technical systems: Building technologies
  for communities}.
\newblock Interaction Design Foundation, 2013.

\bibitem{route_choice_withAV}
{\sc Wong, T.~W., Saxena, N., and Dixit, V.~V.}
\newblock A study of route choice behavior of drivers in autonomous vehicles.
\newblock Tech. rep., 2018.

\bibitem{eval_CPT}
{\sc Wu, Y., Xu, C., and Zhang, T.}
\newblock Evaluation of renewable power sources using a fuzzy mcdm based on
  cumulative prospect theory: A case in china.
\newblock {\em Energy 147\/} (2018), 1227--1239.

\bibitem{xiao_CACC_micro}
{\sc Xiao, L., Wang, M., and van Arem, B.}
\newblock Realistic car-following models for microscopic simulation of adaptive
  and cooperative adaptive cruise control vehicles.
\newblock {\em Transportation Research Record 2623}, 1 (2017), 1--9.

\bibitem{cloud_defense_CPT}
{\sc Xu, D., Xiao, L., Mandayam, N.~B., and Poor, H.~V.}
\newblock Cumulative prospect theoretic study of a cloud storage defense game
  against advanced persistent threats.
\newblock In {\em 2017 IEEE Conference on Computer Communications Workshops
  (INFOCOM WKSHPS)\/} (2017), IEEE, pp.~541--546.

\bibitem{xu_routeCPT}
{\sc Xu, H., Zhou, J., and Xu, W.}
\newblock A decision-making rule for modeling travelers’ route choice
  behavior based on cumulative prospect theory.
\newblock {\em Transportation Research Part C: Emerging Technologies 19}, 2
  (2011), 218--228.

\bibitem{yu2018noncoop}
{\sc Yu, H., Tseng, H.~E., and Langari, R.}
\newblock A human-like game theory-based controller for automatic lane
  changing.
\newblock {\em Transportation Research Part C: Emerging Technologies 88\/}
  (2018), 140--158.

\bibitem{zhang_frnd_cumulative}
{\sc Zhang, C., Liu, T.-L., Huang, H.-J., and Chen, J.}
\newblock A cumulative prospect theory approach to commuters’ day-to-day
  route-choice modeling with friends’ travel information.
\newblock {\em Transportation Research Part C: Emerging Technologies 86\/}
  (2018), 527--548.

\end{thebibliography}

\end{document}